%% file: main.tex
\documentclass[envcountsame,envcountsect,runningheads]{llncs}

\input{packages}

\input{macros}

\begin{document}

\title{A Coinductive Treatment of Infinitary Rewriting}

\author{J\"{o}rg~Endrullis\inst{1} \and 
  Helle~Hvid~Hansen\inst{2,3} \and
  Dimitri~Hendriks\inst{1} \and\\
  Andrew Polonsky\inst{1} \and
  Alexandra Silva\inst{2}
}
   
\authorrunning{Endrullis \and Hansen \and Hendriks \and Polonsky \and Silva}

\institute{VU University Amsterdam, Department of Computer Science
  \and
  Radboud University Nijmegen, Department of Computer Science
  \and
  Centrum Wiskunde \& Informatica
  \and
  VU University Amsterdam, Department of Computer Science
}

\maketitle

\begin{abstract}
  \input{abstract}
\end{abstract}

\section{Introduction}\label{sec:intro}
\input{intro}

\section{Preliminaries}\label{sec:prelims}
\input{prelims}

\section{Introduction to Coinduction}\label{sec:coinduction}
\input{coinduction}

\section{Infinitary Equational Reasoning}\label{sec:ieq}
\input{equational}

\section{Infinitary Term Rewriting}\label{sec:itrs}
\input{itrs}

\section{A Formalization in Coq}\label{sec:discussion}
\input{discussion}

\section{Equivalence with the Standard Definition}\label{sec:equivalence}
\input{equivalence}




\section{Conclusion}
\input{conclusion}

\bibliographystyle{plain}
\bibliography{main}

\end{document}

%% file: packages.tex
\usepackage[colorlinks,linkcolor=dblue,citecolor=dblue,urlcolor=dblue]{hyperref}
\usepackage{amsmath,amssymb,amstext}
\usepackage{enumerate}
\usepackage{proof}
\usepackage{wrapfig}
\usepackage{framed}
\usepackage{atbeginend}
\usepackage{tikz}
\usetikzlibrary{arrows,automata,backgrounds,calc,chains,fadings,folding,decorations.fractals,fit,patterns,positioning,mindmap,shadows,shapes.geometric,shapes.symbols,through,trees,decorations.markings,decorations.text}
\colorlet{dblue}{blue!40!black}

%% file: macros.tex
\newcommand{\comment}[1]{}

\newcommand{\msf}{\mathsf}
\renewcommand{\mit}{\mathit}

\newcommand{\ap}{{\scriptsize@}}

\newcommand{\nix}[1]{}

\newtheorem{notation}[theorem]{Notation}

\AfterBegin{definition}{\normalfont}
\AfterBegin{notation}{\normalfont}

\newcommand{\pairlft}{{\langle}}
\newcommand{\pairrgt}{{\rangle}}
\newcommand{\pairsep}{{,\,}}
\newcommand{\pairstr}[1]{\pairlft#1\pairrgt}
\newcommand{\pair}[2]{\pairstr{#1\pairsep#2}}

\newcommand{\sBNFis}{{{:}{:}{=}}}

\newcommand{\BNFor}{\mathrel{|}}
\newcommand{\coBNFis}{\mathrel{\sBNFis^{\text{co}}}}

\newcommand{\atrs}{\mathcal{R}}

\newcommand{\sred}{{\rightarrow}}
\newcommand{\red}{\mathrel{\sred}}
\newcommand{\rerat}[2]{\mathrel{\sred_{#1,#2}}}

\newcommand{\ieq}{=^\infty}
\newcommand{\ieqdown}{\stackrel{\makebox(0,0){$\smallsmile$}}{=}^\infty}
\newcommand{\ired}{\mathrel{\mbox{$\to\hspace{-2.8mm}\to\hspace{-2.8mm}\to$}}}
\newcommand{\iredi}{\mathrel{\reflectbox{$\ired$}}}
\newcommand{\ireddown}{\rightharpoondown\hspace{-2.8mm}\rightharpoondown\hspace{-2.8mm}\rightharpoondown}
\newcommand{\ireddownfin}{\rightharpoondown\hspace{-2.8mm}\stackrel{<}{\rightharpoondown}\hspace{-2.8mm}\rightharpoondown}
\newcommand{\rstep}{\to_{\varepsilon}}

\newcommand{\ibi}{\mathrel{{}^\infty\!\!\to^\infty}}
\newcommand{\ibidown}{\mathrel{{}^\infty\!\!\stackrel{\makebox(0,0){$\smallsmile$}}{\to}^\infty}}

\newcommand{\redord}{\to_{\textit{ord}}}
\newcommand{\iredord}{\ired_{\textit{ord}}}

\newcommand{\iredordr}[1]{\ired_{\textit{ord},#1}}
\newcommand{\iredorddown}{\ireddown_{\textit{ord}}}

\newcommand{\fun}[1]{\mathsf{#1}}

\newcommand{\relcomp}{\circ}

\newcommand{\nat}{\mathbb{N}}

\newcommand{\sarity}{\#}
\newcommand{\arity}[1]{\sarity(#1)}
\newcommand{\avars}{\mathcal{X}}
\newcommand{\vars}[1]{\mathcal{V}\hspace{-.1ex}\mit{ar}(#1)}

\newcommand{\ster}{\mathit{Ter}}
\newcommand{\ter}[2]{\ster(#1,#2)}
\newcommand{\siter}{\ster^{\infty}}
\newcommand{\iter}[2]{\siter(#1,#2)}

\newcommand{\ametric}{d}
\newcommand{\metric}[2]{\ametric(#1,#2)}

\newcommand{\posemp}{\epsilon}
\newcommand{\apos}{p}

\newcommand{\pos}[1]{\mathcal{P}\!os(#1)}

\newcommand{\subtrm}[2]{#1|_{#2}}
\newcommand{\asubst}{\sigma}

\newcommand{\slfp}{\mu}
\newcommand{\lfp}[2]{\slfp {#1}.\,#2}
\newcommand{\sgfp}{\nu}
\newcommand{\gfp}[2]{\sgfp {#1}.\,#2}

\newcommand{\down}[1]{\overline{#1}}

\newcommand{\id}{\msf{id}}

\newcommand{\tlat}{L}

\newcommand{\Pow}{\mathcal{P}}

\newcommand{\rsplit}{\ensuremath{\msf{split}}}
\newcommand{\rlift}{\ensuremath{\msf{lift}}}
\newcommand{\rid}{\ensuremath{\msf{id}}}

\newcommand{\nest}{\textit{nest}}
\newcommand{\fp}{\textit{fp}}

\newcommand{\boxx}[1]{\colorbox[rgb]{0.99,0.78,0.07}{\kern0.15em#1\kern0.15em}\quad}

\newcommand{\hole}{\raisebox{-2pt}{\scalebox{.7}[1.5]{$\Box$}}}

%% file: abstract.tex
We present a coinductive treatment 
of infinitary term rewriting with reductions of arbitrary ordinal length.
Our framework allows the following succinct 
definition of the infinitary rewrite relation $\ired$:
\begin{align*}
  {\ired} \;\;=\;\; \lfp{x}{\gfp{y}{(\rstep \cup \mathrel{\down{x}})^*\relcomp \down{y}}}
\end{align*}
where $\down{R} = \{\,\pair{f(s_1,\ldots,s_n)}{\,f(t_1,\ldots,t_n)} \mid s_1 \mathrel{R} t_1,\ldots,s_n \mathrel{R} t_n\,\} \,\cup\, \id$
and $\rstep$ are root steps.
Here $\slfp$ is the least fixed point operator and $\sgfp$ is the greatest 
fixed point operator.
\medskip

In contrast to the usual definitions of infinitary rewriting,
our setup has neither need for ordinals nor for metric convergence.
This makes the framework especially suitable for formalizations in theorem provers.
On the basis of the above definition we  provide a proof of the Compression Lemma in the Coq theorem prover.
\medskip

Finally, we present our coinductive framework in the form of coinductive proof rules,
giving rise to proof terms for infinite reductions. 

%% file: intro.tex
Infinitary rewriting is a generalization of the ordinary finitary rewriting
to infinite terms and infinite reductions (including reductions of ordinal lengths larger than $\omega$).
We present a coinductive treatment of infinitary rewriting
free of ordinals, metric convergence and partial orders which have been
essential in earlier definitions of the 
concept~\cite{ders:kapl:plai:1991,kenn:1992,kenn:klop:slee:vrie:1995a,Dezani-CiancagliniSV03,klop:vrij:2005,kenn:klop:slee:vrie:1997,kenn:vrie:2003,KennawaySSV05,kahr:2013,bahr:2010,bahr:2010b,bahr:2012,endr:hend:klop:2012}.

Let us give the idea. 
Let $R$ be a term rewriting system (TRS).
We write $\rstep$ for root steps with respect to $R$\,:
${\rstep} = \{\,(\ell\sigma,r\sigma) \mid \ell \to r \in R,\; \text{$\sigma$ a substitution}\,\}$.
The crucial ingredient of our definition of infinitary rewriting $\ired$ are the \emph{coinductive} rules
\begin{align}
  \begin{aligned}
  \infer=
  {s \ired t}
  {s \mathrel{(\rstep \cup \ireddown)^*} t}
  &\hspace{2cm}&
  \infer=
  {f(s_1,s_2,\ldots,s_n) \ireddown f(t_1,t_2,\ldots,t_n)}
  {s_1 \ired t_1 & \ldots & s_n \ired t_n}
  \end{aligned}
  \label{rules:main}
\end{align}
Here $\ired$ and $\ireddown$ 
stand for finite and infinite reductions
where $\ireddown$ contains only steps below the root.
The coinductive nature of the rules means that the proof terms 
need not be well-founded.
%

\begin{example}\label{ex:a:Ca}
  Let $R$ be the TRS consisting solely of the following rewrite rule
  {\abovedisplayskip.75ex 
   \belowdisplayskip.75ex
  \begin{align*}
    \fun{a} \to \fun{C}(\fun{a})
  \end{align*}}%
  We write $\fun{C}^\omega$ to denote the infinite term $\fun{C}(\fun{C}(\fun{C}(\ldots)))$,
  the solution of the equation $\fun{C}^\omega = \fun{C}(\fun{C}^\omega)$.
  We then have $\fun{a} \ired \fun{C}^\omega$, that is, an infinite reduction from $\fun{a}$ to $\fun{C}^\omega$ in the limit:
  {\abovedisplayskip-.5ex 
   \belowdisplayskip.75ex
  \begin{align*}
    \fun{a} \to \fun{C}(\fun{a}) \to \fun{C}(\fun{C}(\fun{a})) \to \fun{C}(\fun{C}(\fun{C}(\fun{a}))) \to \ldots 
    \to^\omega \fun{C}^\omega
  \end{align*}}%
  Using the rules above, we can derive $\fun{a} \ired \fun{C}^\omega$ as shown in Figure~\ref{fig:aComega}.
  \noindent
  This is an infinite proof tree as indicated by the loop 
  \raisebox{.5ex}{\tikz \draw [->,thick,dotted] (0,0) -- (5mm,0mm);}
  in which the rewrite sequence 
  $\fun{a} \rstep \fun{C}(\fun{a}) \ireddown \fun{C}^\omega$
  is written in the form
  $\fun{a} \rstep \fun{C}(\fun{a}) \quad \fun{C}(\fun{a})\ireddown \fun{C}^\omega$, that is,
  two separate steps such that the target of the first equals the source of the second step;
  this is made precise in Notation~\ref{not:transitivity}, below.
\end{example}\vspace{-2ex}

  \begin{wrapfigure}{r}{4.5cm}
    \vspace{-7ex}
    \begin{framed}
    \vspace{-1.5ex}
    \begin{align*}
      \infer=
      {\fun{a} \ired \fun{C}^\omega}
      {
        \fun{a} \rstep \fun{C}(\fun{a})
        &
        \infer=
        {\fun{C}(\fun{a}) \ireddown \fun{C}^\omega}
        {\infer={\fun{a} \ired \fun{C}^\omega}
            {\makebox(0,0){
              \hspace{13mm}\begin{tikzpicture}[baseline=11ex]
              \draw [->,thick,dotted] (0,0) -- (0,1mm) to[out=90,in=100] (11mm,-1mm) to[out=-80,in=-20,looseness=1.4] (-0mm,-13mm);
              \end{tikzpicture}
            }}
        }
      }
    \end{align*}
    \vspace{-6ex}
    \end{framed}
    \vspace{-4ex}
    \caption{A reduction $\fun{a} \ired \fun{C}^\omega$.}
    \vspace{-3ex}
    \label{fig:aComega}
  \end{wrapfigure}
  Put in words, the proof tree in Figure~\ref{fig:aComega} can be described as follows.
  We have an infinitary rewrite sequence~$\ired$ from $\fun{a}$ to $\fun{C}^\omega$
  since we have a root step from $\fun{a}$ to $\fun{C}(\fun{a})$, and
  an infinitary reduction below the root $\ireddown$ from $\fun{C}(\fun{a})$ to $\fun{C^\omega}$.
  The latter reduction $\fun{C}(\fun{a}) \ireddown \fun{C^\omega}$ is in turn witnessed
  by the infinitary rewrite sequence $\fun{a} \ired \fun{C}^\omega$ 
  on the direct subterms.

\begin{notation}\label{not:transitivity}
  Instead of introducing derivation rules for transitivity,
  in particular for
  $\mathrel{(\rstep \cup \ireddown)^*}$,
  we write rewrite sequences 
  $s_0 \rightsquigarrow_0 s_1 \rightsquigarrow_1 \ldots \rightsquigarrow_{n-1} s_n$
  where ${\rightsquigarrow_i} \in \{\rstep, \ireddown\}$
  as sequence of single steps: 
  {\abovedisplayskip2ex 
   \belowdisplayskip.75ex
  \begin{align*}
    \infer=
    {s_0 \ired s_n}
    {s_0 \rightsquigarrow_0 s_1 \quad s_1 \rightsquigarrow_1 s_2 \quad\ldots\quad s_{n-1} \rightsquigarrow_{n-1} s_n}
  \end{align*}}%
  \noindent
  This notation is more convenient since it avoids the need for 
  explicitly introducing rules for transitivity,
  and thereby keeps the proof trees small.
\end{notation}

As a second example, let us consider a rewrite sequence of length beyond $\omega$.

\begin{example}\label{ex:fab}
  We consider the term rewriting system with the following rules:
  {\abovedisplayskip.75ex 
   \belowdisplayskip.75ex
  \begin{align*}
    \fun{f}(x,x) &\to \fun{D} & 
    \fun{a} &\to \fun{C}(\fun{a}) &
    \fun{b} &\to \fun{C}(\fun{b})
  \end{align*}}%
  Then we have the following reduction of length $\omega+1$:
  {\abovedisplayskip.75ex 
   \belowdisplayskip.75ex
  \begin{align*}
    \fun{f}(\fun{a},\fun{b}) 
    \to \fun{f}(\fun{C}(\fun{a}),\fun{b}) 
    \to \fun{f}(\fun{C}(\fun{a}),\fun{C}(\fun{b}))
    \to \ldots \to^\omega \fun{f}(\fun{C}^\omega,\fun{C}^\omega)
    \to \fun{D}
  \end{align*}}%
  That is, after an infinite rewrite sequence of length $\omega$, 
  we reach the limit term $\fun{f}(\fun{C}^\omega,\fun{C}^\omega)$,
  and we then continue with a rewrite step from $\fun{f}(\fun{C}^\omega,\fun{C}^\omega)$ to $\fun{D}$.
\end{example}
  \vspace{-2ex}

  \begin{wrapfigure}{r}{7.4cm}
    \vspace{-7ex}
    \begin{framed}
    \vspace{-1.5ex}
    \begin{align*}
      \infer=
      {\fun{f}(\fun{a},\fun{b}) \ired \fun{D}}
      {
        \infer=
        {\fun{f}(\fun{a},\fun{b}) \ireddownfin \fun{f}(\fun{C}^\omega,\fun{C}^\omega)}
        {
          \infer=
          {\fun{a} \ired \fun{C}^\omega}
          {\text{like Figure~\ref{fig:aComega}}}
          &
          \infer=
          {\fun{b} \ired \fun{C}^\omega}
          {\text{like Figure~\ref{fig:aComega}}}
        }
        & 
        \fun{f}(\fun{C}^\omega,\fun{C}^\omega) \rstep \fun{D}
      }
    \end{align*}
    \vspace{-5.5ex}
    \end{framed}
    \vspace{-4ex}
    \caption{A reduction $\fun{f}(\fun{a},\fun{b}) \ired \fun{D}$.}
    \vspace{-7ex}
    \label{fig:fab}
  \end{wrapfigure}

  Figure~\ref{fig:fab} shows how this rewrite sequence \mbox{$\fun{f}(\fun{a},\fun{b}) \ired \fun{D}$}
  can be derived in our setup.
  The precise meaning of the symbol $\ireddownfin$ in the figure will
  be explained later; 
  for the moment, we may think of $\ireddownfin$ to be $\ireddown$.
  We note that the rewrite sequence $\fun{f}(\fun{a},\fun{b}) \ired \fun{D}$
  cannot be `compressed' to length $\omega$. That is,
  there exists no reduction $\fun{f}(\fun{a},\fun{b}) \to^{\le \omega} \fun{D}$.

For the definition of rewrite sequences of ordinal length,
there is a design choice concerning the connectedness at limit ordinals:
(a) metric convergence, or (b) strong convergence.
The purpose of the connectedness condition is to exclude jumps at limit ordinals, 
as illustrated in the non-connected rewrite sequence 
  {\abovedisplayskip.75ex 
   \belowdisplayskip.75ex
  \begin{align*}
    \underbrace{\fun{a} \to \fun{a} \to \fun{a} \to \ldots}_{\text{$\omega$-many steps}} \;\fun{b} \to \fun{b}
  \end{align*}}%
where $R = \{\,\fun{a} \to \fun{a},\, \fun{b}\to\fun{b}\,\}$.
The rewrite sequence stays $\omega$ steps at $\fun{a}$ and in the limit step `jumps' to $\fun{b}$.

The connectedness condition with respect to \emph{metric convergence} requires
that for every limit ordinal~$\gamma$,
the terms $t_\alpha$ converge with limit $t_\gamma$ as $\alpha$ approaches $\gamma$ from below. 
The \emph{strong convergence} requires additionally that
the depth of the rewrite steps $t_{\alpha} \to t_{\alpha+1}$ tends to infinity 
as $\alpha$ approaches $\gamma$ from below.
The standard notion of infinitary rewriting~\cite{tere:2003,endr:hend:klop:2012} is based on strong convergence
as it gives rise to a more elegant rewriting theory;
for example, allowing to trace symbols and redexes over limit ordinals.
This is the notion that we are concerned with in this paper.

The rules~\eqref{rules:main} give rise to infinitary rewrite sequences in a very natural way,
without the need for ordinals, metric convergence, or depth requirements.
The depth requirement in the definition of strong convergence
arises naturally in the rules~\eqref{rules:main} by employing coinduction over the term structure.
Indeed, it is not difficult to see that the coinductive rules~\eqref{rules:main} capture all 
infinitary strongly convergent reductions $s \ired t$.
This is a consequence of a result due to~\cite{kenn:klop:slee:vrie:1995a}
which states that every strongly convergent rewrite sequence
contains only a finite number of steps at any depth $d \in \nat$
and in particular only a finite number of root~steps.
Hence every strongly convergent reduction is of the form ${(\ireddown \relcomp \rstep)^*} \relcomp \ireddown$.

While this argument shows that every strongly convergent reduction \mbox{$s \ired t$}
can be derived using the rules~\eqref{rules:main},
it does not guarantee that we can derive precisely the strongly convergent reductions.
Actually, the rules do allow to derive more, as the following example shows.
\begin{example}\label{ex:Ca:a}
  Let $R$ consist of the rewrite rule
  $\fun{C}(\fun{a}) \to \fun{a}$.
  Using the rules~\eqref{rules:main}, we can derive 
  $\fun{C}^\omega \ired \fun{a}$ as shown in Figure~\ref{fig:backwards}.
\end{example}
  \vspace{-1.5ex}

  \begin{wrapfigure}{r}{5cm}
    \vspace{-8ex}
    \begin{framed}
    \vspace{-1.5ex}
    \begin{align*}
      \infer=
      {\fun{C}^\omega \ired \fun{a}}
      {
        \infer=
        {\fun{C}^\omega \ireddownfin \fun{C}(\fun{a})}
        {\infer={\fun{C}^\omega \ired \fun{a}}
            {\makebox(0,0){
              \hspace{-13mm}\begin{tikzpicture}[baseline=12ex]
              \draw [->,thick,dotted] (0,0) -- (0,1mm) to[out=90,in=80] (-11mm,-1mm) to[out=-100,in=-160,looseness=1.6] (-0mm,-13mm);
              \end{tikzpicture}
            }}
        }
        &
        \fun{C}(\fun{a}) \rstep \fun{a}
      }
    \end{align*}
    \vspace{-6ex}
    \end{framed}
    \vspace{-4ex}
    \caption{Derivation of $\fun{C}^\omega \ired \fun{a}$.}
    \vspace{-5ex}
    \label{fig:backwards}
  \end{wrapfigure}
  With respect to the standard notion of infinitary rewriting $\ired$ in the literature
  we do not have $\fun{C}^\omega \ired \fun{a}$ since $\fun{C}^\omega$ is a normal form
  (does not contain an occurrence of the left-hand side $\fun{C}(\fun{a})$ of the rule).
  Note that the rule $\fun{C}(\fun{x}) \to \fun{x}$ also gives rise to $\fun{C}^\omega \ired \fun{a}$
  by the same derivation as in Figure~\ref{fig:backwards}.

This example illustrates that, without further restrictions, the rules~\eqref{rules:main}
give rise to a notion of infinitary rewriting
that allows rewrite sequences to extend infinitely forwards, but also infinitely backwards;
we call this \emph{bi-infinite rewriting}.
Here backwards does \emph{not} refer to reversing the arrow $\leftarrow_{\varepsilon}$.
By replacing $\rstep$ with $\leftarrow_{\varepsilon} \cup \rstep$ in the first rule,
we obtain a theory of \emph{infinitary equational reasoning}.
This notion of infinitary equational reasoning
has the property of strong convergence built in,
and thereby allows to trace redex occurrences forwards as well as backwards.
Due to space limitations, we leave the investigation of these concepts to future work.

The focus of this paper is the standard notion of infinitary rewriting.
\emph{How to obtain the strongly convergent rewrite sequences $s \ired t$?}
For this purpose it suffices to impose a syntactic restriction on the shape of the proof trees 
obtained from the rules~\eqref{rules:main}.
The idea is that all rewrite sequences $\ireddown$ in $(\rstep \cup \ireddown)^*$, that
are before a root step $\rstep$, should be shorter than the rewrite sequence that we are defining.
To this end, we change $(\rstep \cup \ireddown)^*$ to $\mathrel{(\rstep \cup \ireddownfin)^*} \relcomp \ireddown$
where $\ireddownfin$ is a marked equivalent of $\ireddown$,
and we employ the marker to exclude infinite nesting of $\ireddownfin$.
Then we have an infinitary strongly convergent rewrite sequence from $s$ to $t$ 
if and only if $s \ired t$ can be derived by the rules
\begin{align}
  \begin{aligned}
  \infer=
  {s \ired t}
  {s \mathrel{(\rstep \cup \ireddownfin)^*} \relcomp \ireddown t}
  &\hspace{.4cm}&  
  \infer=
  {f(s_1,s_2,\ldots,s_n) \stackrel{(<)}{\ireddown} f(t_1,t_2,\ldots,t_n)}
  {s_1 \ired t_1 & \ldots & s_n \ired t_n}
  &\hspace{.4cm}&
  \infer=
  {s \stackrel{(<)}{\ireddown} s}
  {}
  \end{aligned}
  \label{rules:restrict}
\end{align}
in a (not necessarily well-founded) proof tree without infinite nesting of $\ireddownfin$.
In other words,  we only allow those proof trees 
in which all paths (ascending through the proof tree) contain only
finitely many occurrences of $\ireddownfin$.

We note that the second and third rule are abbreviations for two rules
each:
the symbol $\stackrel{(<)}{\ireddown}$ stands for $\ireddown$ and for $\ireddownfin$.
Intuitively, $\ireddownfin$ can be thought of as infinitary rewrite sequence 
below the root that is `smaller' than the sequence we are defining.
Here `smaller' refers to the nesting depth of $\ireddownfin$,
but can equivalently be thought of the length of the reduction (in some well-founded order).
\begin{example}
  Let us revisit Examples~\ref{ex:a:Ca}, \ref{ex:fab} and \ref{ex:Ca:a}.
  Example~\ref{ex:a:Ca} contains no occurrences of $\ireddownfin$.
  The proof tree in Example~\ref{ex:fab} has a single occurrence of
  $\ireddownfin$, but this occurrence is not contained in the indicated loops, and thus not infinitely nested.
  Only Example~\ref{ex:Ca:a} contains a symbol $\ireddownfin$ on a
  loop, and hence a path with infinitely many occurrences of $\ireddownfin$,
  and thus the proof tree is excluded by the syntactic restriction.
  \qed
\end{example}

\subsection*{Related Work}
\input{related}

\subsection*{Outline}
In Section~\ref{sec:prelims} we introduce infinitary rewriting in the usual way
with ordinal-length rewrite sequences, and convergence at every limit ordinal.
We then continue in Section~\ref{sec:coinduction}
with an introduction to coinduction.
We give two definitions of infinitary rewriting based on
mixing induction and coinduction in Section~\ref{sec:itrs}.
In Section~\ref{sec:discussion} we illustrate that our framework
is suitable for formalizations in theorem provers.
In Section~\ref{sec:equivalence},
we prove the equivalence of our coinductive definitions of infinitary rewriting
with the standard definition.

%% file: related.tex
While the basic idea of a coinductive treatment of infinitary rewriting is not new~\cite{coqu:1996,joac:2004,endr:polo:2012b},
%
%
the previous approaches have in common that they do not
capture rewrite sequences of length $> \omega$. 
The coinductive treatment presented here captures all strongly
convergent rewrite sequences of arbitrary ordinal length.

From the topological perspective, various notions of infinitary rewriting 
and infinitary equational reasoning have been studied in~\cite{kahr:2013}.
We note that none of the rewrite notions 
considered in this paper are continuous (forward closed)
in general. Here continuity of $\to$ means that $\lim_{i\to\infty} t_i = t$ and $\forall i.{s \to t_i}$ implies $s \to t$.
However, continuity might hold for certain classes of term rewrite systems;
see further~\cite{endr:hend:klop:2012} for continuity 
in strongly convergent infinitary rewriting $\ired$.


%% file: prelims.tex
We give a brief introduction to infinitary rewriting.
For further reading on infinitary rewriting we refer to
\cite{klop:vrij:2005,tere:2003,bare:klop:2009,endr:hend:klop:2012},
for an introduction to finitary rewriting to
\cite{klop:1992,tere:2003,baad:nipk:1998,bare:1977}.

A \emph{signature $\Sigma$} is a set of symbols $f$ each having a fixed arity $\arity{f} \in \nat$.
Let $\avars$ be an infinite set of variables such that $\avars \cap \Sigma = \emptyset$.
The set of (finite and) \emph{infinite terms $\iter{\Sigma}{\avars}$}
over $\Sigma$ and $\avars$ is \emph{coinductively} (see further~\cite{barr:1993}) defined by the grammar:
$
  T \coBNFis x \BNFor f(\underbrace{T,\ldots,T}_{\text{$\arity{f}$ times}}) \; \text{($x \in \avars$, $f \in \Sigma$)}
  \label{grammar:term}
$
Intuitively, coinductively means that the grammar rules may be applied
an infinite number of times.
The equality on the terms is bisimilarity.
For a brief introduction to coinduction, we refer to Section~\ref{sec:coinduction}.

We define the \emph{identity relation on terms} by $\id = \{\pair{s}{s} \mid s \in \iter{\Sigma}{\avars}\}$.

\begin{remark}
Alternatively, the infinite terms arise from the set of finite terms, $\ter{\Sigma}{\avars}$, 
by metric completion,
using the well-known distance function $d$ such that for $t,s \in \ter{\Sigma}{\avars}$, 
$d(t,s) = 2^{-n}$ if the $n$-th level of the terms $t,s$ (viewed as labeled trees) 
is the first level where a difference appears, 
in case $t$ and $s$ are not identical; furthermore, $d(t,t) = 0$.
It is standard that this construction yields $\pair{\ter{\Sigma}{\avars}}{d}$ as a metric space. 
Now infinite terms are obtained by taking the completion of this metric space, 
and they are represented by infinite trees. 
We will refer to the complete metric space arising in this way as $\pair{\iter{\Sigma}{\avars}}{d}$, 
where $\iter{\Sigma}{\avars}$ is the set of finite and infinite terms over~$\Sigma$.
\end{remark}

Let $t \in \iter{\Sigma}{\avars}$ be a finite or infinite term.
The set of \emph{positions $\pos{t}\subseteq \nat^*$ of $t$} is
defined by: $\posemp \in \pos{t}$ and 
$i\vec{p} \in \pos{t}$ whenever $t = f(t_1,\ldots,t_n)$ 
with $1\le i\le n$ and $\vec{p} \in \pos{t_i}$.
For $p \in \pos{t}$, the \emph{subterm $\subtrm{t}{p}$ of $t$ at position $p$}
is defined by
$\subtrm{t}{\posemp} = t$ and
$\subtrm{f(t_1,\ldots,t_n)}{ip} = \subtrm{t_i}{p}$.
%
The set of \emph{variables $\vars{t}\subseteq \avars$ of $t$} is 
$\vars{t} = \{x \in \avars \mid \exists \, p\in \pos{t}.\,\subtrm{t}{p} = x\}$.

A \emph{substitution $\asubst$} is a map $\asubst : \avars \to \iter{\Sigma}{\avars}$.
We extend the domain of substitutions~$\asubst$ to 
$\iter{\Sigma}{\avars}$ by coinduction, as follows:
$\asubst(f(t_1,\ldots,t_n)) = f(\asubst(t_1),\ldots,\asubst(t_n))$.
For terms $s$ and substitutions $\asubst$, we write $s\sigma$ for $\sigma(s)$.
We write $x \mapsto s$ for the substitution defined by
$\asubst(x) = s$ and $\asubst(y) = y$ for all $y \ne x$.
Let $\hole$ be a fresh variable.
A \emph{context} $C$ is a term $\iter{\Sigma}{\avars \cup \{\hole\}}$
containing precisely one occurrence of the variable $\hole$.
For contexts $C$ and terms $s$ we write $C[s]$ for $C(\hole \mapsto s)$.

A \emph{rewrite rule $\ell \to r$} over $\Sigma$ and $\avars$ is a pair 
$(\ell,r) \in \iter{\Sigma}{\avars} \times \iter{\Sigma}{\avars}$
of terms such that the left-hand side $\ell$ is not a variable ($\ell \not\in \avars$),
and all variables in the right-hand side $r$ occur in $\ell$ ($\vars{r} \subseteq \vars{\ell}$).
Note that we do neither require the left-hand side nor the right-hand side of a rule
to be finite.

A \emph{term rewrite system (TRS) $\atrs$} over $\Sigma$ and $\avars$
is a set of rewrite rules over $\Sigma$ and $\avars$.
A TRS induces a rewrite relation on the set of terms as follows.
For $\apos \in \nat^\ast$ we define ${\rerat{\atrs}{\apos}} \subseteq \iter{\Sigma}{\avars} \times \iter{\Sigma}{\avars}$, a \emph{rewrite step at position $\apos$}, by
$
  C[\ell\sigma] \rerat{\atrs}{\apos} C[r\sigma]
$
if $C$ a context with $\subtrm{C}{\apos} = \hole$,\; $\ell \to r \in \atrs$,\; $\sigma : \avars \to \iter{\Sigma}{\avars}$.
We write $s \to_{\atrs} t$ if $s \rerat{\atrs}{\apos} t$ for some $\apos\in\nat^\ast$.
A \emph{normal form} is a term without a redex occurrence,
that is, a term that is not of the form $C[\ell\sigma]$ for some context $C$, 
rule $\ell \to r\in \atrs$ and substitution $\sigma$.

A natural consequence of this construction is the emergence of the notion of \emph{metric convergence}: 
we say that $t_0 \to t_1 \to t_2 \to \cdots$ is an infinite reduction sequence with limit $t$, 
if $t$ is the limit of the sequence $t_0,t_1,t_2, \ldots$ in the usual sense of metric convergence. 
Metric convergence is sometimes also called \emph{weak convergence}.
In fact, we will use throughout a stronger notion that has better properties. 
This is \emph{strong convergence}, which in addition to the stipulation for metric (or weak) convergence, 
requires that the depth of the redexes contracted in the 
successive steps tends to infinity when approaching a
limit ordinal from below.
So this rules out the possibility that the action of redex contraction stays confined at the top, 
or stagnates at some finite level of depth. 


\begin{definition}\label{def:itrs:standard}
  A \emph{transfinite rewrite sequence} (of ordinal length $\alpha$)
  is a sequence of rewrite steps 
  $(t_{\beta} \rerat{\atrs}{\apos_{\beta}} t_{\beta+1})_{\beta < \alpha}$
  such that for every limit ordinal $\lambda < \alpha$ we have that 
  if $\beta$ approaches $\lambda$ from below, then
  \begin{enumerate}[(i)]
    \item\label{item:distance}
      the distance $\metric{t_\beta}{t_\lambda}$ tends to $0$ 
      and, moreover,
    \item\label{item:depth}
      the depth of the rewrite action, i.e., 
      the length of the position $\apos_\beta$, 
      tends to infinity.
  \end{enumerate}

  The sequence is called \emph{strongly convergent} 
  if $\alpha$ is a successor ordinal, 
  or there exists a term $t_\alpha$ such that
  the conditions~\ref{item:distance} and~\ref{item:depth}
  are fulfilled for every limit ordinal $\lambda \leq \alpha$.
  In this case we write $t_0\iredordr{\atrs} t_\alpha$, 
  or $t_0\red^{\alpha} t_\alpha$ 
  to explicitly indicate the length $\alpha$ of the sequence.
  The sequence is called \emph{divergent} if it is not strongly convergent.
\end{definition}

There are several reasons why strong convergence is beneficial; 
the foremost being that in this way we can define the notion of \emph{descendant} 
(also \emph{residual}) over limit ordinals. 
Also the well-known Parallel Moves Lemma
and the Compression Lemma
fail for weak convergence, see~\cite{simo:2004} and \cite{ders:kapl:plai:1991} respectively.

%% file: coinduction.tex
We briefly introduce the relevant concepts from
(co)algebra and (co)induction that will be used later throughout this
paper. 
For a more thorough introduction, we refer to \cite{jaco:rutt:2011}.
There will be two main points where coinduction will play a role, in the definition of terms and in the definition 
of the term rewriting. 

Terms are usually defined with respect to a signature $\Sigma$. 
For instance, consider the type of lists with elements in a given set $A$. 
\begin{verbatim}
  type List a = Empty | Cons a (List a)
\end{verbatim}
The above grammar corresponds to the signature (or type constructor)
$\Sigma(X) = 1 + A \times X$ where the $1$ is used as a placeholder for the empty list {\tt Empty} and the second component represents the {\tt Cons} constructor.
Such a grammar can be interpreted in two ways: 
The \emph{inductive} interpretation yields as terms the set of finite lists,
and corresponds to the \emph{least fixed point} of $\Sigma$.
The \emph{coinductive} interpretation yields as terms
the set of all finite or infinite lists,
and corresponds to the \emph{greatest fixed point} of $\Sigma$.
More generally, the inductive interpretation of a signature 
yields finite terms (with well-founded syntax trees), and
dually, the coinductive interpretation 
yields possibly infinite terms.
For readers familiar with the categorical definitions of algebras
and coalgebras, these two interpretations amount to defining
finite terms as the \emph{initial $\Sigma$-algebra}, 
and possibly infinite terms as the \emph{final $\Sigma$-coalgebra}.

Equality on finite terms is the expected syntactic/inductive definition.  
Equality of possibly infinite terms is observational equivalence (or bisimilarity). For instance, in the above example, two infinite lists $\sigma$ and $\tau$ are equal if and only if they are related by a List-bisimulation. A relation $R\subseteq \mathtt{List\ a} \times \mathtt{List\ a}$ is a List-bisimulation if and only if for all pairs $(\mathtt{Cons\ a\ }\sigma, \mathtt{Cons\ b\ }\tau)\in R$, it holds that $a=b$ and $(\sigma,\tau) \in R$.

Formally, term rewriting is a relation on a set $T$ of terms,
and hence an element of the complete lattice $\tlat := \Pow(T \times T)$,
i.e., the powerset of $T \times T$.
Relations on terms can thus be defined using least and greatest fixed points
of monotone operators on $\tlat$.
In this setting, an inductively defined relation is 
a least fixed point $\mu F$ of a monotone $F \colon \tlat \to \tlat$;
and dually, a coinductively defined relation is
a greatest fixed point $\nu F$ of a monotone $F \colon \tlat \to \tlat$.
These notions of induction and coinduction are, in fact, also instances of 
the more abstract categorical definitions. This can be seen by viewing 
$\tlat$ as a partial order (ordered by set inclusion). In turn, a partial order
$(P,\leq)$ can be seen as a category whose objects are the elements of $P$ 
and there is a unique arrow $X \to Y$ if $X \leq Y$. 
A functor on $(P,\leq)$ is then nothing but a monotone map $F$;
an $F$-coalgebra $X \to F(X)$ is a post-fixed point of $F$; and
a final $F$-coalgebra is a greatest fixed point of $F$.
The existence of the final $F$-coalgebra is guaranteed by 
the Knaster-Tarski fixed point theorem.
Coinduction, and similarly induction, can now be formulated as proof rules:
\begin{equation}\label{eq:coind-ind-rules}
\frac{X \leq F(X)}{X \leq \nu F}(\nu\text{-rule}) \qquad\qquad 
\frac{F(X) \leq X}{\mu F \leq X}(\mu\text{-rule})
\end{equation}
that express the fact that $\nu F$ is the greatest post-fixed point of $F$,
and $\mu F$ is the least pre-fixed point of $F$.

%% file: equational.tex
From the basic rules~\eqref{rules:main} 
arises in a very natural way a novel notion of infinitary equational reasoning $\ieq$.
%
This notion is the natural counterpart of strongly convergent infinitary rewriting.
Like infinitary strongly convergent reductions,
the theory of infinitary equational reasoning has the property that
every derivation contains only a finite number of reasoning steps at any depth $d \in \nat$.
We consider an \emph{equational specification (ES)} as a TRS.

\begin{samepage}
\begin{definition}\label{def:ieq:rules}
  Let $E$ be an equational specification over $\Sigma$.
  We define \emph{infinitary equational reasoning} \emph{${=^\infty} \subseteq T \times T$} 
  on terms $T = \iter{\Sigma}{\avars}$ by the following coinductive rules
  \begin{align*}
    \infer=
    {s \ieq t}
    {s \mathrel{(\leftarrow_{\varepsilon} \cup \rstep \cup \ieqdown)^*} t}
    &&  
    \infer=
    {f(t_1,t_2,\ldots,t_n) \ieqdown f(t'_1,t'_2,\ldots,t'_n)}
    {t_1 \ieq t'_1 & \ldots & t_n \ieq t'_n}
  \end{align*}
  where ${\ieqdown} \subseteq T \times T$ stands for infinitary equational reasoning below the root.
  \qed
\end{definition}
\end{samepage}
  
\begin{example}\label{ex:ieq}
  Let $E$ be an equational specification consisting of the equations (rules):
  \begin{align*}
    \fun{a} &= \fun{f}(\fun{a})  &
    \fun{b} &= \fun{f}(\fun{b}) &
    \fun{C}(\fun{b}) &= \fun{C}(\fun{C}(\fun{a})) &
  \end{align*}
  Then $\fun{a} \ieq \fun{b}$ as derived in Figure~\ref{fig:ieq:fagb} (top),
  and $\fun{C}(\fun{a}) \ieq \fun{C}^\omega$ as in Figure~\ref{fig:ieq:fagb} (bottom).
  \begin{figure}[h]
    \begin{gather*}
    \infer=
    { \fun{a} \ieq \fun{b} }
    {
      \fun{a} \to_{\varepsilon} \fun{f}(\fun{a})
      &
      \infer=
      {\fun{f}(\fun{a}) \ieqdown \fun{f}^\omega}
      {
        \infer=
        {\fun{a} \ieq \fun{f}^\omega}
        {
          \fun{a} \to_{\varepsilon} \fun{f}(\fun{a})
          &
          \infer=
          {\fun{f}(\fun{a}) \ieqdown \fun{f}^\omega}
          {\infer={\fun{a} \ieq \fun{f}^\omega}
              {\makebox(0,0){
                \hspace{13mm}\begin{tikzpicture}[baseline=11ex]
                \draw [->,thick,dotted] (0,0) -- (0,1mm) to[out=90,in=100] (11mm,-1mm) to[out=-80,in=-20,looseness=1.4] (-0mm,-13mm);
                \end{tikzpicture}
              }}
          }
        }
      }
      &&
      \infer=
      {\fun{f}^\omega \ieqdown \fun{f}(\fun{b})} 
      {
        \infer=
        {\fun{f}^\omega \ieq \fun{b}}
        {
          \infer=
          {\fun{f}^\omega \ieqdown \fun{f}(\fun{b})}
          {\infer={\fun{f}^\omega \ieq \fun{b}}
              {\makebox(0,0){
                \hspace{-13mm}\begin{tikzpicture}[baseline=11ex]
                \draw [->,thick,dotted] (0,0) -- (0,1mm) to[out=90,in=80] (-10mm,-1mm) to[out=-100,in=-160,looseness=1.6] (-0mm,-12mm);
                \end{tikzpicture}
              }}
          }
          &
          \fun{f}(\fun{b}) \leftarrow_{\varepsilon} \fun{b}
        }
      }
      &
      \fun{f}(\fun{b}) \leftarrow_{\varepsilon} \fun{b}
    }
    \\
    \infer=
    {\fun{C}(\fun{a}) \ieq \fun{C}^\omega}
    {
      \infer=
      {\fun{C}(\fun{a}) \ieqdown \fun{C}(\fun{b})} 
      {\infer={a \ieq b}{\text{(as above)}}}
      &
      \fun{C}(\fun{b}) \rstep \fun{C}(\fun{C}(\fun{a}))
      &
      \infer=
      {\fun{C}(\fun{C}(\fun{a})) \ieqdown \fun{C}^\omega}
      {\infer={\fun{C}(\fun{a}) \ieq \fun{C}^\omega}
            {\makebox(0,0){
              \hspace{18mm}\begin{tikzpicture}[baseline=12.5ex]
              \draw [->,thick,dotted] (0,0) -- (0,1mm) to[out=90,in=100] (15mm,-1mm) to[out=-80,in=-25,looseness=1.4] (-0mm,-14mm);
              \end{tikzpicture}
            }}
      }
    }
    \end{gather*}\vspace{-5ex}
    \caption{Infinitary equational reasoning.}
    \label{fig:ieq:fagb}
  \end{figure}
  \\It is easy to see that ${(\iredi \relcomp \ired)^*} \subseteq {\ieq}$,
  and $\fun{C}(\fun{a}) \ieq \fun{C}^\omega$ shows that this inclusion is strict.
\end{example}

Definition~\ref{def:ieq:rules} of $\ieq$ can be equivalently be defined using 
a greatest fixed point as follows.

\begin{definition}\label{def:ieq:fixedpoint}
  Let $E$ be an equational specification over $\Sigma$,
  and $T= \iter{\Sigma}{\avars}$.
  For $R \in \Pow(T \times T)$, we define its \emph{lifting} as
  \begin{align*}
    \down{R} &= \{\,\pair{f(s_1,\ldots,s_n)}{\,f(t_1,\ldots,t_n)} \mid s_1 \mathrel{R} t_1,\ldots,s_n \mathrel{R} t_n\,\}
                \,\cup\, \id 
  \end{align*}
  We define the relation $\ieq$ as $\gfp{x}{(\leftarrow_{\varepsilon} \cup \rstep \cup \mathrel{\down{x}})^*}$.
  \qed
\end{definition}

It is easy to verify that the function 
$x \mapsto (\leftarrow_{\varepsilon} \cup \rstep \cup \mathrel{\down{x}})^*$
is monotone, and consequently the greatest fixed point in Definition~\ref{def:ieq:fixedpoint}
exists. 

Another notion that arises naturally in our setup
is that of bi-infinite rewriting,
allowing rewrite sequences to extend infinitely forwards and backwards.
We emphasize that each of the steps $\rstep$ 
in such sequences is a forward step.

\begin{definition}\label{def:ibi:rules}
  Let $R$ be a term rewriting system over $\Sigma$,
  and let $T= \iter{\Sigma}{\avars}$.
  We define \emph{bi-infinite rewrite relation} \emph{${\ibi} \subseteq T \times T$} 
  by the following coinductive rules
  \begin{align*}
    \infer=
    {s \ibi t}
    {s \mathrel{(\rstep \cup \ibidown)^*} t}
    &&  
    \infer=
    {f(t_1,t_2,\ldots,t_n) \ibidown f(t'_1,t'_2,\ldots,t'_n)}
    {t_1 \ibi t'_1 & \ldots & t_n \ibi t'_n}
  \end{align*}
  where ${\ibidown} \subseteq T \times T$ stands for bi-infinite rewriting below the root.
  \qed
\end{definition}

Examples~\ref{ex:a:Ca},~\ref{ex:fab} and~\ref{ex:Ca:a} are illustrations of this rewrite relation.
Note that these examples employ the symbols $\ired$ and $\ireddown$ instead of $\ibi$ and $\ibidown$, respectively.
In the infinitary conversion in Example~\ref{ex:ieq} 
we need to reverse the rule $\fun{b} = \fun{f}(\fun{b})$
in order to obtain a bi-infinite rewrite sequence $\fun{a} \ibi \fun{b}$.

%% file: itrs.tex
We present two -- ultimately equivalent -- definitions of infinitary rewriting \mbox{$s \ired t$},
based on mixing induction and coinduction.
We summarize the definitions:
\begin{enumerate}[A.]
  \item \emph{Derivation Rules.}
        First, we define $s \ired t$ via a syntactic restriction on the proof trees 
        that arise from the coinductive rules~\eqref{rules:restrict}.
        The restriction excludes all proof trees that contain ascending paths
        with an infinite number of marked symbols.
  \medskip
  \item \emph{Mixed Induction and Coinduction.}
        Second, we define $s \ired t$ based on mutually mixing induction and coinduction,
        that is, least fixed points $\mu$ and greatest fixed points $\nu$.
\end{enumerate}
\noindent
In contrast to previous coinductive definitions~\cite{coqu:1996,joac:2004,endr:polo:2012b}, 
the setup proposed here captures all strongly convergent rewrite sequences (of arbitrary ordinal length).

Throughout this section, we fix a signature $\Sigma$ and a term
rewriting system $\atrs$ over $\Sigma$. The notation $\rstep$ denotes a root step with respect to $\atrs$.

\subsection{Derivation Rules}
The first definition has already been discussed in the introduction.
The strongly convergent rewrite sequences are obtained by 
a syntactic restriction on the formation of the proof trees that arise from rules~\eqref{rules:main}.

\begin{definition}\label{def:ired:restrict}
  We define the relation \emph{$\ired$} on terms $T= \iter{\Sigma}{\avars}$ as follows.
  We have $s \ired t$ if there exists a 
  (finite or infinite) proof tree $\delta$ deriving $s \ired t$ using the rules
  \begin{align*}
    \infer=[\rsplit]
    {s \ired t}
    {s \mathrel{(\rstep \cup \ireddownfin)^*} \relcomp \ireddown t}
    &&  
    \infer=[\rlift]
    {f(s_1,s_2,\ldots,s_n) \stackrel{(<)}{\ireddown} f(t_1,t_2,\ldots,t_n)}
    {s_1 \ired t_1 & \ldots & s_n \ired t_n}
    &&
    \infer=[\rid]
    {s \stackrel{(<)}{\ireddown} s}
    {}
  \end{align*}
  such that $\delta$ does not contain infinite nesting\footnote{%
    No infinite nesting of $\ireddownfin$ means that there exists no path ascending 
    through the proof tree that meets an infinite number of symbols $\ireddownfin$.
  } %
  of $\ireddownfin$. 
  The symbol $\stackrel{(<)}{\ireddown}$ stands for $\ireddown$ or $\ireddownfin$;
  so the second rule is an abbreviation for two rules; similarly for the third rule.
  \qed
\end{definition}

Let us give some intuition for the rules in Definition~\ref{def:ired:restrict}.
The relation $\ireddownfin$ can be thought of as 
an infinitary reduction below the root, 
that is `shorter' than the reduction that we are deriving.
The three rules (\rsplit, \rlift\ and \rid) can be interpreted as follows:
\begin{enumerate}[(i)]
  \item 
    The \rsplit-rule:
    the term $s$ rewrites infinitarily to $t$, $s \ired t$, 
    if $s$ rewrites to $t$ using a finite sequence of (a) root steps,
    and (b) infinitary reductions $\ireddown$ below the root
    (where infinitary reductions preceding root steps must be shorter than the derived reduction).
    \smallskip
  \item 
    The \rlift-rule: 
    the term $s$ rewrites infinitarily to $t$ below the root, $s \stackrel{(<)}{\ireddown} t$,
    if the terms are of the shape $s = f(t_1,t_2,\ldots,t_n)$ and $t = f(t'_1,t'_2,\ldots,t'_n)$
    and there exist reductions $\ired$ on the arguments:
    $t_1 \ired t'_1$, \ldots, $t_n \ired t'_n$.
    \smallskip
  \item The \rid-rule
    allows the rewrite relation $\stackrel{(<)}{\ireddown}$ to be reflexive,
    and this in turn yields reflexivity of $\ired$.
    For variable-free terms, reflexivity can already be derived using the first two rules.
    However, for terms with variables, this third rule is needed 
    (unless we treat variables as constant symbols).
    \smallskip
\end{enumerate}
For example proof trees using the rules from Definition~\ref{def:ired:restrict}, 
we refer to Examples~\ref{ex:a:Ca} and~\ref{ex:fab} in the introduction. 

\subsection{Mixed Induction and Coinduction}

%
The next definition is based on mixing induction and coinduction. 
The inductive part is used to model the restriction 
to finite nesting of $\ireddownfin$ in the proofs in Definition~\ref{def:ired:restrict}.
The induction corresponds to a least fixed point $\slfp$,
while a coinductive rule to a greatest fixed point $\sgfp$.

\begin{definition}\label{def:ired:fixedpoint}
  Let $T= \iter{\Sigma}{\avars}$ be the set of terms,
  and let $L$ be the set of all relations on terms $L = \Pow(T\times T)$.
  For $R \in \Pow(T \times T)$, we define its \emph{lifting}~as
  \begin{align*}
    \down{R} &= \{\,\pair{f(s_1,\ldots,s_n)}{\,f(t_1,\ldots,t_n)} \mid s_1 \mathrel{R} t_1,\ldots,s_n \mathrel{R} t_n\,\}
                \,\cup\, \id 
  \end{align*}
  We define the relation $\ired$ by
  \begin{equation*}
      {\ired} \quad = \quad \lfp{x}{\gfp{y}{(\rstep \cup \mathrel{\down{x}})^*\relcomp \down{y}}}\\[-5ex]
  \end{equation*}
  \qed
\end{definition}
Define functions $G \colon L \times L \to L$ and $F \colon L \to L$ by
\begin{equation*}
  G(x,y) = (\rstep \cup \mathrel{\down{x}})^*\relcomp (\down{y})
    \quad\text{ and }\quad
  F(x) = \gfp{y}{G(x,y)} = \gfp{y}{(\rstep \cup \mathrel{\down{x}})^*\relcomp (\down{y})}
\end{equation*}
Then  
    ${\ired} \;=\; \lfp{x}{F(x)} \;=\; \lfp{x}{\gfp{y}{G(x,y)}} \;=\; 
    \lfp{x}{\gfp{y}{(\rstep \cup \mathrel{\down{x}})^*\relcomp \down{y}}}$
It can easily be verified that $F$ and $G$ are monotone (in all their arguments).
Recall that a function $f$ over sets is monotone if $X \subseteq Y \implies f(\ldots,X,\ldots) \subseteq f(\ldots,Y,\ldots)$.
Hence $F$ and $G$ have unique least and greatest fixed points.

The reflexive, transitive closure $(\cdot)^*$ in Definition~\ref{def:ired:fixedpoint} can, of course,
also be defined using a least fixed point, for example, as follows:
\begin{align*}
  R^* &= \lfp{z}{(\id \cup R  \relcomp z}) 
  &\text{or equivalently }&&
  R^* &=\lfp{z}{(\id \cup R  \cup z \relcomp z})
\end{align*}
Unfolding this definition of the reflexive, transitive closure in Definition~\ref{def:ired:fixedpoint}
we obtain:
\(
{\ired} = \lfp{x}{\gfp{y}{(\lfp{z}{\id \cup (\rstep \cup \mathrel{\down{x}}) \relcomp z})\relcomp \down{y}}}
\)\;.

\paragraph*{Comparing Definitions~\ref{def:ired:restrict} and~\ref{def:ired:fixedpoint}}

We emphasize the close connection between Definitions~\ref{def:ired:restrict} and~\ref{def:ired:fixedpoint}.
Observe that the clause $(\rstep \cup \mathrel{\down{x}})^*\relcomp (\down{y})$ 
in Definition~\ref{def:ired:fixedpoint} 
models $\mathrel{(\rstep \cup \ireddownfin)^*} \relcomp \ireddown$ 
in the first rule of Definition~\ref{def:ired:restrict}.
Here $\down{x}$ corresponds to $\ireddownfin$, and $\down{y}$ to $\ireddown$.
The least fixed point $\slfp{x}$ caters for the restriction of the proof tree formation
to finite nesting of $\ireddownfin$. 

Definitions~\ref{def:ired:restrict} and~\ref{def:ired:fixedpoint}
of the rewrite relation $\ired$ both have their merits. 
Definition~\ref{def:ired:fixedpoint}, which is based on mixing induction and coinduction, 
is a succinct, mathematically precise formulation of $\ired$. 
The derivation rules, Definition~\ref{def:ired:restrict},
on the other hand, are easy to understand, and easy to use for humans.

%% file: discussion.tex
The standard definition of infinitary rewriting,
using ordinal length rewrite sequences and strong convergence at limit ordinals,
is difficult to formalize. 
The coinductive framework we propose, is easy to formalize and work with in theorem provers.
For example, in Coq, the coinductive definition of 
infinitary strongly convergent reductions can be defined as follows:
{\small
\begin{verbatim}
Inductive ired : relation term :=
  | Ired : 
      forall R I : relation term,
      subrel I ired ->
      subrel R ((root_step (+) lift I)* ;; lift R) ->
      subrel R ired.
\end{verbatim}}
\noindent
Here 
\verb=term= is the set of coinductively defined terms,
\verb=;;= is relation composition,
\verb=(+)= is the union of relations, 
\verb=*= the reflexive transitive closure,
\verb=lift R= is~$\down{R}$,
and \verb=root_step= is the root step relation.

Let us briefly comment on this formalization.
We have ${\ired} \;=\; \lfp{x}{\gfp{y}{G(x,y)}}$
where $G(x,y) = (\rstep \cup \mathrel{\down{x}})^*\relcomp \down{y}$.
The inductive definition of \verb=ired= corresponds to the least fixed point $\slfp{x}$.
Coq has no support for mutual inductive and coinductive definitions.
Therefore, instead of the explicit coinduction, we use the $\nu$-rule from~\eqref{eq:coind-ind-rules}.
For every relation $R$
that fulfills $R \subseteq  G(x,R)$, we have that $R \subseteq \gfp{y}{G(x,y)}$.
Moreover, we know that $\gfp{y}{G(x,y)}$ is the union of all these relations $R$.
Finally, we introduce an auxiliary relation \verb=I= 
to help Coq generate a good induction principle.
One can think of \verb=I= as consisting of those pairs
for which the recursive call to \verb=ired= is invoked.
Replacing \verb=lift I= by \verb=lift ired= is correct, 
but then the induction principle that Coq generates for \verb=ired= is useless.

On the basis of the above definition we proved the Compression Lemma in Coq,
that is, we have proven that if $s \ired t$ in a left-linear TRS, then $s \to^{\le \omega} t$.
To the best of our knowledge this is the first formal proof of this well-known lemma.
The formalization is available at \url{http://www.cs.vu.nl/~diem/coq/compression/}.

%% file: equivalence.tex
In this section we prove the equivalence of the coinductively defined 
infinitary rewrite relations $\ired$ from
Definitions~\ref{def:ired:restrict}, and~\ref{def:ired:fixedpoint}
with the standard definition
based on ordinal length rewrite sequences with metric and strong convergence at every limit ordinal
(Definition~\ref{def:itrs:standard}).

\subsection{Derivation Rules}

Let $\ired$ be the relation defined in Definition~\ref{def:ired:restrict}.
The definition requires that the nesting structure of $\ireddownfin$
in proof trees is well-founded. As a consequence, we can associate to every proof tree
a (countable) ordinal that allows to embed the nesting structure in an order-preserving way.
We use $\omega_1$ to denote the first uncountable ordinal,
and we view ordinals as the set of all smaller ordinals
(then the elements of $\omega_1$ are all countable ordinals).

\begin{definition}
  Let $\delta$ be a proof tree as in Definition~\ref{def:ired:restrict},
  and let $\alpha$ be an ordinal.
  An \emph{$\alpha$-labeling of $\delta$} 
  is a labeling of all symbols $\ireddownfin$ in $\delta$ with elements from $\alpha$
  such that
  each label is strictly greater than all labels occurring in the subtrees (all labels above).
  \qed
\end{definition}

\begin{lemma}\label{lem:nest}
  Every proof tree as in Definition~\ref{def:ired:restrict}
  has an $\alpha$-labeling for some $\alpha \in \omega_1$.
  \qed
\end{lemma}


\begin{definition}
  Let $\delta$ be a proof tree as in Definition~\ref{def:ired:restrict}.
  We define the \emph{nesting depth} of $\delta$ as 
  the least ordinal $\alpha \in \omega_1$ such that $\delta$ admits an $\alpha$-labeling.
  For every $\alpha \le \omega_1$, we define a relation
  ${\ired_\alpha} \subseteq {\ired}$
  as follows:
  $s \ired_\alpha t$ whenever $s \ired t$
  can be derived using a proof with nesting depth $< \alpha$.
  Likewise we define relations
  ${\ireddown_\alpha}$ and
  ${\ireddownfin_\alpha}$.
  \qed
\end{definition}

As a direct consequence of Lemma~\ref{lem:nest} we have:
\begin{corollary}\label{cor:nest:omega1}
  We have ${\ired_{\omega_1}} = {\ired}$.
  \qed
\end{corollary}

We will now show that the coinductively defined 
infinitary rewrite relation $\ired$ (Definition~\ref{def:ired:restrict})
coincides with the standard definition of $\iredord$ (Definition~\ref{def:itrs:standard})
based on ordinal length rewrite sequences with metric and strong convergence at every limit ordinal.
The crucial observation is the following theorem from~\cite{klop:vrij:2005}:
\begin{theorem}[Theorem 2 of~\cite{klop:vrij:2005}]\label{thm:finite}
  A transfinite reduction is divergent if and only if for some $N$
  there are infinitely many steps at depth $N$.
\end{theorem}

We are now ready to prove the equivalence of both notions:
\begin{theorem}
  We have ${\ired} = {\iredord}$.
\end{theorem}

\newcommand{\redd}{\rightharpoondown}
\begin{proof}
  We write $\redd$ for steps that are not at the root,
  and $\iredorddown$ to denote a reduction $\ireddown$ without root steps.

  We begin with the direction ${\iredord} \subseteq {\ired}$.
  We show by induction on the ordinal length $\alpha$
  that we have both ${\redord^\alpha} \subseteq {\ired}$ and ${\redd^\alpha} \subseteq {\stackrel{(<)}{\ireddown}}$.
  Let $\alpha$ be an ordinal and $s,t$ terms.
%
  We proceed by coinduction on the structure of the proof tree to derive $\ired$:  
  \begin{enumerate}[(i)]
    \item 
      Assume that $s \redord^\alpha t$, that is, we have a strongly convergent reduction $\sigma$ from $s$ to $t$
      of length $\alpha$.
      By Theorem~\ref{thm:finite} the rewrite sequence $\sigma$
      contains only a finite number of root steps.
      As a consequence, $\sigma$ is of the form:
      $s \mathrel{(\rstep \cup \redd^{<\alpha})^* \circ \redd^{\le \alpha}} t$.
      Note that the reductions $\iredorddown$ preceding root steps must be shorter
      than $\alpha$ since the last root step is contracted at an index $<\alpha$ in the reduction $\sigma$.
      By induction hypothesis we have ${\redd^{<\alpha}} \subseteq {\ireddownfin}$. 
      Then $s \mathrel{(\rstep \cup \ireddownfin)^* \circ \redd^{\le \alpha}} t$.
      Hence, $s \ired t$ can be derived using the \rsplit-rule
      since by coinduction hypothesis we have ${\redd^{\le \alpha}} \subseteq {\ired}$.
      Observe that the thereby constructed proof tree for $s \ired t$
      contains no infinite nesting of $\ireddownfin$ because every marker $\ireddownfin$
      occurs in a node where the induction hypothesis has been applied.  
      An infinite nesting of markers would thus give rise to an infinite descending chain of ordinals, 
      which is impossible by well-foundedness of $\alpha$.
    \item 
      Assume that $s \redd^\alpha t$, that is, we have a strongly convergent reduction $\sigma$ 
      without root steps from $s$ to $t$ of length $\alpha$.
      Then the terms $s,t$ must be of the shape
      $s = \fun{f}(s_1,\ldots,s_n)$ and $t = \fun{f}(t_1,\ldots,t_n)$,
      and $\sigma$ can be split in reductions $s_1 \redord^{\le \alpha} t_1$, \ldots, $s_n \redord^{\le \alpha} t_n$ on the arguments.
      By (i) we have $s_1 \ired t_1$, \ldots, $s_n \ired t_n$.
      Hence by the \rlift-rule we obtain $s \ireddown t$ and $s \ireddownfin t$
      (the nesting of $\ireddownfin$ stays well-founded).
  \end{enumerate}

  \noindent
  We now show ${\ired} \subseteq {\iredord}$.
  We prove by well-founded induction on $\alpha \le \omega_1$ that
  ${\ired_{\alpha}} \subseteq {\iredord}$.
  This suffices since ${\ired} = {\ired_{\omega_1}}$.
  Let $\alpha \le \omega_1$ and assume that $s \ired_{\alpha} t$.
  Let $\delta$ be a proof tree of nesting depth $\le \alpha$ deriving $s \ired_{\alpha} t$.
  The only possibility to derive $s \ired t$ is an application of the \rsplit-rule
  with the premise $s \mathrel{(\rstep \cup \ireddownfin)^*} \relcomp \ireddown t$.
  Since $s \ired_{\alpha} t$, we have
  $s \mathrel{(\rstep \cup \ireddownfin_{\alpha})^*} \relcomp \ireddown_{\alpha} t$.
  By induction hypothesis we have 
  $s \mathrel{(\rstep \cup \iredord)^*} \relcomp \ireddown_{\alpha} t$,
  and thus
  $s \iredord \relcomp \ireddown_{\alpha} t$.
  We have ${\ireddown_{\alpha}} = {\down{\ired_{\alpha}\vphantom{i}}}$, and consequently
  $s \iredord s_1 \mathrel{\down{\ired_{\alpha}\vphantom{i}}} t$ for some term $s_1$.  
  Repeating this argument on
  $s_1 \mathrel{\down{\ired_{\alpha}\vphantom{i}}} t$, we get 
  $s \iredord s_1 \mathrel{\down{\iredord\vphantom{i}}} s_2 \mathrel{\down{\down{\ired_{\alpha}\vphantom{i}}}} t$.  
  After $n$ iterations, we obtain
  \begin{align*}
    s \iredord s_1 
    \mathrel{\down{\iredord\vphantom{i}}} s_2 
    \mathrel{\down{\down{\iredord\vphantom{i}}}} s_3
    \mathrel{\down{\down{\down{\iredord\vphantom{i}}}}} s_4
    \cdots \mathrel{({\ired_{\alpha}})^{-(n-1)}} s_n
    \mathrel{({\ired_{\alpha}})^{-n}} t
  \end{align*}
  where $({\ired_{\alpha}})^{-n}$
  denotes the $n$'th iteration of $x \mapsto \down{x}$ on $\ired_{\alpha}$.
  
  Clearly, the limit of $\{s_n\}$ is $t$.  Furthermore, each of the reductions $s_n \iredord s_{n+1}$ 
  are strongly convergent and take place at depth greater than or equal to $n$.
  Thus, the infinite concatenation of these reductions yields a strongly convergent reduction from $s$ to $t$
  (there is only a finite number of rewrite steps at any depth $n$).
\end{proof}

\subsection{Mixed Induction and Coinduction}

\begin{theorem}
  The Definitions~\ref{def:ired:restrict} and~\ref{def:ired:fixedpoint} 
  give rise to the same relation ${\ired}$.
\end{theorem}

\begin{proof}
  To avoid confusion we write $\ired_{\nest}$ for the relation $\ired$ defined in Definition~\ref{def:ired:restrict},
  and $\ired_{\fp}$ for the relation $\ired$ defined in Definition~\ref{def:ired:fixedpoint}.
  We show ${\ired_{\nest}} = {\ired_{\fp}}$.

  We begin with ${\ired_{\fp}} \subseteq {\ired_{\nest}}$.
  Employing the $\mu$-rule from~\eqref{eq:coind-ind-rules},
  it suffices to show that $F({\ired_{\nest}}) \subseteq {\ired_{\nest}}$.
  We prove this fact by coinduction on the structure of coinductively defined proof trees (Definition~\ref{def:ired:restrict}).
  We have ${\ireddown_{\nest}} = {\ireddownfin_{\nest}} = \down{\ired_{\nest}\vphantom{i}}$, and thus 
  \begin{align*}
    F({\ired_{\nest}}) &= (\rstep \cup \mathrel{\down{\ired_{\nest}\vphantom{i}}})^* \relcomp \down{F({\ired_{\nest}})} 
    = (\rstep \cup \mathrel{\ireddownfin_{\nest}})^* \relcomp \down{F({\ired_{\nest}})}\\
    \down{F({\ired_{\nest}})} &= \id \cup \{\,\pair{f(\vec{s})}{\,f(\vec{t})} \mid \vec{s} \,\mathrel{F(\ired_{\nest})}\, \vec{t}\,\}
  \end{align*}
  where $\vec{s}$, $\vec{t}$ abbreviate $s_1,\ldots,s_n$ and $t_1,\ldots,t_n$, respectively,
  and we write $\vec{s} \mathrel{R} \vec{t}$
  if we have $s_1 \mathrel{R} t_1,\ldots,s_n \mathrel{R} t_n$.
  Now we apply the \rsplit-rule to derive 
  $(\rstep \cup \mathrel{\ireddownfin_{\nest}})^* \relcomp \down{F({\ired_{\nest}})}$
  and 
  $\down{F({\ired_{\nest}})}$
  can be derived via the \rid-rule, or the \rlift-rule;
  for the arguments $\vec{s}$, $\vec{t}$ of the \rlift-rule 
  we have by coinduction that $\vec{s} \ired_{\nest} \vec{t}$ since $\vec{s} \,\mathrel{F(\ired_{\nest})}\, \vec{t}$.
  
  We now show that ${\ired_{\nest}} \subseteq {\ired_{\fp}}$.
  We prove by well-founded induction on $\alpha \le \omega_1$ that
  ${\ired_{\alpha,\nest}} \subseteq {\ired_{\fp}}$.
  This yields the claim ${\ired_{\omega_1,\nest}} = {\ired_{\nest}}$ by Corollary~\ref{cor:nest:omega1}.
  Since $\ired_{\fp}$ is a fixed point of $F$,
  we obtain ${\ired_{\fp}} = F(\ired_{\fp})$, and since $F(\ired_{\fp})$ 
  is a greatest fixed point,
  using the $\nu$-rule from~\eqref{eq:coind-ind-rules},
  it suffices to show that $(*)$ ${\ired_{\alpha,\nest}} \subseteq G(\ired_{\fp},\ired_{\alpha,\nest})$.
  Thus assume that $s \ired_{\alpha,\nest} t$,
  and let $\delta$ be a proof tree of nesting height $\le \alpha$ deriving $s \ired_{\alpha,\nest} t$.
  The only possibility to derive $s \ired_{\nest} t$ is an application of the \rsplit-rule
  with the premise $s \mathrel{(\rstep \cup \ireddownfin_{\nest})^*} \relcomp \ireddown_{\nest} t$.
  Since $s \ired_{\alpha,\nest} t$, we have
  $s \mathrel{(\rstep \cup \ireddownfin_{\alpha,\nest})^*} \relcomp \ireddown_{\alpha,\nest} t$.
  Let $\tau$ be one of the steps $\ireddownfin_{\alpha,\nest}$ displayed in the premise.
  Let $u$ be the source of $\tau$ and $v$ the target,
  so $\tau : u \ireddownfin_{\alpha,\nest} v$.
  The step $\tau$ is derived either via the \rid-rule or the \rlift-rule.
  The case of the \rid-rule is not interesting since we then can drop $\tau$ from the premise.
  Thus let the step $\tau$ be derived using the \rlift-rule.
  Then the terms $u,v$ are of form $u = f(u_1,\ldots,u_n)$ and $v = f(v_1,\ldots,v_n)$
  and for every $1 \le i \le n$ we have $u_i \ired_{\beta,\nest} v_i$ for some $\beta < \alpha$.
  Thus by induction hypothesis we obtain $u_i \ired_{\fp} v_i$ for every $1 \le i \le n$,
  and consequently $u \mathrel{\down{\ired_{\fp}\vphantom{i}}} v$.
  We then have $s \mathrel{(\rstep \cup \mathrel{\down{\ired_{\fp}\vphantom{i}})^*}} \relcomp \ireddown_{\alpha,\nest} t$,
  and hence $s \mathrel{G(\ired_{\fp},\ired_{\alpha,\nest})} t$.
  This concludes the proof.
\end{proof}

%% file: conclusion.tex
We have proposed a coinductive treatment of infinitary rewriting
which, in contrast to previous coinductive treatments, captures
the full infinitary term rewriting with rewrite sequences of arbitrary ordinal length.
 
We summarise a few of the merits of our coinductive framework:
\begin{enumerate}
  \item We establish a bridge between infinitary rewriting and coalgebra.
    Both fields are concerned with infinite objects and it is interesting
    to understand their relation better.%
\medskip
  \item We give a succinct, mathematically precise definition of infinitary rewriting.%
\medskip
  \item The framework paves the way for formalizing infinitary rewriting in theorem provers
    (as illustrated by our proof of the Compression Lemma in Coq).
\medskip
  \item The coinductive derivation rules establish proof terms for infinite reductions.%
\medskip
  \item 
    From our framework arise two natural variants of infinitary rewriting
    that we believe are new
    (recall that 
    ${\ired}$ is defined by ${\ired} = \lfp{x}{\gfp{y}{(\rstep \cup \mathrel{\down{x}})^*\relcomp \down{y}}}$): 
    \begin{enumerate}[(a)]
    \smallskip
      \item bi-infinite rewriting
          ${\ibi} \;=\; \gfp{y}{(\rstep \cup \mathrel{\down{y}})}$, and
    \smallskip
      \item infinitary equational reasoning
          ${\ieq} \;=\; \gfp{y}{(\leftarrow_{\varepsilon} \cup \rstep \cup \mathrel{\down{y}})}$. 
    \medskip
    \end{enumerate}
    As a consequence of the coinduction over the term structure,
    these notions have the strong convergence built in,
    and thus can profit from the well-developed techniques (such as tracing)
    in infinitary rewriting.
\medskip
  \item Our work is also a case study on mixed inductive/coinductive definitions.%
\end{enumerate}
Concerning the proof terms for infinite reductions, let us mention that 
an alternative approach has been developed in parallel by Lombardi, R\'{\i}os and de~Vrijer~\cite{lomb:ros:vrij:2013}.
While we focus on proof terms for the reduction relation, 
they use proof terms for modeling the fine-structure of 
the infinite reductions themselves.

Our work lays the foundation for several directions of future research:
\begin{enumerate}
  \item The revealed connection between infinitary rewriting and coalgebra
    provides the basis for a deeper study of the relation of both fields. 
\medskip
  \item 
    The concepts of bi-infinite rewriting and infinitary equational reasoning are novel.
    It is interesting to study these concepts,
    in particular since the theory of infinitary equational reasoning is still underdeveloped.
    For example, it would be interesting to compare the 
    Church-Rosser properties 
    \begin{align*}
      {\ieq} \subseteq {\ired \relcomp \iredi} &&\text{ and }&& {(\iredi \relcomp \ired)^*} \subseteq {\ired \relcomp \iredi}
    \end{align*}
    In the extended version~\cite{endr:hans:hend:polo:silv:2013} of the present paper, 
    we have shown that ${\ieq} \subsetneq {(\ired \cup \iredi)^*}$.%
\medskip
  \item 
    The formalization of our framework in Coq and the proof of the Compression Lemma 
    are but the first steps towards the formalization of all major theorems in infinitary rewriting.
\medskip
  \item It is interesting to investigate whether and how the coinductive framework
    can be extended to other notions of infinitary rewriting,
    for example reductions 
    where root-active terms are mapped to $\bot$ in the limit~\cite{bahr:2010,bahr:2010b,bahr:2012,endr:hend:klop:2012}.
\end{enumerate}

\subsection*{Acknowledgments}
We thank Patrick Bahr and Jeroen Ketema for fruitful discussions 
and comments to earlier versions of this paper.